\documentclass{birkmult}
\usepackage{graphicx}
\usepackage{cite}
 \newtheorem{thm}{Theorem}[section]
 
 \newtheorem{lem}[thm]{Lemma}
 \newtheorem{prop}[thm]{Proposition}
 \theoremstyle{definition}
 
 \theoremstyle{remark}
 \newtheorem{rem}[thm]{Remark}
 \newtheorem{rems}[thm]{Remarks}
 
 \newtheorem*{exs}{Examples}
 \numberwithin{equation}{section}

\newcommand{\be}{\begin{equation}}
\newcommand{\ee}{\end{equation}}
\newcommand{\bea}{\begin{eqnarray}}
\newcommand{\eea}{\end{eqnarray}}
\newcommand{\beax}{\begin{eqnarray*}}
\newcommand{\eeax}{\end{eqnarray*}}
\newcommand{\mfr}[2]{{\textstyle\frac{#1}{#2}}}

\def\a{\alpha}
\def\b{\beta}
\def\g{\gamma}
\def\d{\delta}
\def\e{{\rm{e}}}
\def\eps{\varepsilon}
\newcommand{\E}{{\mathord{\mathbb E}}}
\def\i{{\rm{i}}}
\def\s{\sigma}

\def\o{\omega}
\def\D{\Delta}
\def\O{\Omega}
\def\l{\lambda}
\def\L{\Lambda}

\def\C{\mathbb C}
\def\Im{{\rm{Im}}}
\def\Re{{\rm{Re}}}
\def\H{\mathbb H}
\def\N{\mathbb N}
\def\R{\mathbb R}
\def\SH{\mathbb S\mathbb H}
\def\Z{\mathbb Z}
\def\cd{{\rm{cd}}}

\begin{document}
%
%
%
%
%
%
%
%
%
\title[Absolutely continuous spectrum]
 {A geometric approach to absolutely continuous spectrum for discrete 
 Schr\"odinger operators}
\author[Richard Froese]{University of British Columbia}
\address{Department of Mathematics\\
Vancouver, British Columbia, Canada}
\email{rfroese@math.ubc.ca}
\author[David Hasler]{College of William \& Mary}
\address{Department of Mathematics\\
Williamsburg, Virginia, USA}
\email{dghasler@wm.edu}
\author[Wolfgang Spitzer]{FernUniversit\"at Hagen}
\address{%
Fakult\"at f\"ur Mathematik und Informatik\\
Hagen, Germany}
\email{Wolfgang.Spitzer@FernUni-Hagen.de}


\date{January 19, 2010}

\begin{abstract}
We review a geometric approach to proving absolutely continuous (ac) spectrum
for random and deterministic Schr\"odinger operators developed in
\cite{FHS1,FHS2,FHS3,FHS4}. We study decaying potentials in one dimension and 
present a simplified proof of ac spectrum of the Anderson model on trees. The
latter implies ac spectrum for a percolation model on trees. Finally, we 
introduce certain loop tree models which lead to some interesting open problems.
\end{abstract}

\maketitle

\section{Introduction}

The study of one-particle Schr\"odinger operators of the form $H=-\D + q$ with 
kinetic energy $-\D$ and (random) potential $q$ has caught the attention of many 
reseachers over several decades. As an introduction to this topic we recommend 
the books by Cycon, Froese, Kirsch, Simon\cite{CFKS}, by Stollmann~\cite{Stoll},
and the paper by Kirsch~\cite{Kirsch}. 

In the discrete setting, we choose the kinetic energy to be the negative of 
the adjacency matrix, $\D$, of some graph $\mathcal G$. The most important 
example is the $d$-dimensional regular graph, $\Z^d$. Since there are only very 
few examples of potentials, $q$, where the spectrum of $H$ is known explicitly 
we would be content knowing, for instance, the existence of point and absolutely 
continuous (henceforth ac) spectrum of $H$, the level statistics of eigenvalues 
or the long-time behavior under the Schr\"odinger time evolution. For example, 
from scattering theory it is well-known that if $q$ decays fast enough (that is,
if $q$ is integrable) then the spectrum of $H=-\D + q$ inside the spectrum of 
$-\D$ (on $\Z^d$, this is the interval $[-2d,2d]$) is purely ac and outside 
this interval the spectrum is pure point. 

An important model in solid state physics concerns the case when $q$ is a
random potential. In the simplest scenario we assume that the values $q(v)$ 
and $q(w)$ for two different vertices $v,w\in\Z^d$ are chosen independently 
from an a-priori given probability measure, $\nu$. Let us multiply the 
potential, $q$, by the factor $a>0$ and interprete $a$ as the \textbf{disorder} 
parameter. Anderson discovered in 1958 that for large disorder or at large
energy the spectrum of $H_a=-\D+a\,q$ is pure point. By now there is an extensive 
literature on this phenomenon which is known as Anderson \textbf{localization}.
The proofs are based on the seminal work of Fr\"ohlich and Spencer~\cite{FS} 
and of Aizenman and Molchanov~\cite{AM}. However, there is currently no proof 
of the existence of ac spectrum at small disorder (or \textbf{delocalization}) 
on $\Z^3$. This is considered an outstanding open problem in Mathematical 
Physics, also known as the \textbf{extended states conjecture}. 

One valuable contribution to this conjecture might come from replacing the graph
$\Z^d$ by a simpler graph such as a tree and study there extended states 
(synonymous with ac spectrum) for random potentials. This has indeed been 
achieved first by Klein~\cite{Klein} in 1998 (and later by Aizenman, Sims,
Warzel\cite{ASW1} in 2006) who proved the extended states conjecture on trees. 
Motivated by Klein's result we first constructed novel examples of potentials on a
tree that produce ac spectrum\cite{FHS1}. Then we reproved a variant of Klein's result 
\cite{FHS2}. A simplified version of this proof is presented in 
Section \ref{tree}. In order to move somewhat closer to the
lattice $\Z^d$ we consider a random potential on a tree that is strongly
correlated instead of independently distributed\cite{FHS3}. We prove that for small
correlations (a large part of) the ac spectrum is stable but it is well-known 
that it disappears completely at maximum correlation, see Section 
\ref{strongly correlated}. In Section \ref{loop tree} we present three
models where we add loops to a (binary) tree. It is only the mean-field loop
tree model where we can solve the spectrum of the new Laplacian. On top of this
Laplacian we add a certain random potential and prove stability of a large ac 
component. After a short review of some spectral theory we discuss 
one-dimensional Schr\"odinger operators. We reprove the stability of the ac 
spectrum with respect to an integrable potential, a Mourre estimate, and the 
stability with respect to a square integrable random potential. The proofs
follow from simple geometric properties of the M\"obius transformation (or
transfer matrix) with respect to the Poincar\'e metric, which controls the 
spreading of the Green function in terms of the potential. In Section 
\ref{higher dimension}, these M\"obius transformations are generalized to 
general graphs (including, for instance, $\Z^d$ and trees), and, like in one
dimension, express the Green function as a limit of products of M\"obius 
transformations. 

\section{Setup}\label{setup}

A graph $\mathcal G=(V,E)$ consists here of a countably infinite set $V$ 
called the vertex set. $E\subseteq V\times V$ is called the 
edge set und obeys
\begin{enumerate}
\item[(i)] if $(v,w)\in E$ then $(w,v)\in E$;
\item[(ii)] $\sup_{v\in V} |\{w\in V:(v,w)\in E\}|<\infty$.
\end{enumerate}
$v,w\in V$ are called nearest neighbors if $(v,w)\in E$. 

The most important example is the $d$-dimensional regular lattice, $\Z^d$, but 
we may as well consider the graph with vertex set $V=\N_0^d$ ($d\in\N$) and 
edge set $E=\{(x,y)\in V\times V:\|x-y\|_1:=\sum_{i=1}^d |x_i-y_i|=1\}$. 
Another example of interest is the (rooted) regular tree, $T_k,k\in\N$. Here, 
$V=\bigcup_{n\ge0,0\le j\le k^n-1}\{(n,j)\}\subset \N_0^2$. 
Two vertices $v=(n,j)$ and $w=(m,\ell)$ are nearest neighbors if 
$m=n+1$ and $j\in\mfr1k\{\ell,\ell+1,\ldots, \ell+k-1\}$ or if $n=m+1$ and 
$\ell\in\mfr1k\{j,j+1,\ldots, j+k-1\}$. The vertex $0$ is called the root.

The graph $\mathcal G = (V,E)$ determines the \textbf{adjacency matrix} 
(operator) on $\ell^2(V)$ of the graph $\mathcal G$ with kernel $\D(v,w)$ 
given by
\be \D(v,w):=\left\{\begin{array}{cc}1&\mbox{ if }(v,w)\in E\\0&\mbox{else}
\end{array}\right.\,.
\ee
That is, for $\phi\in\ell^2(V)$,
\be\label{eq1}
(\D \phi)(v):=\sum_{w\in V} \D(v,w) \phi(w) = \sum_{w\in V:(v,w)\in E}
\phi(w)\,,\quad v\in V\,.
\ee
Because of the two conditions (i) and (ii) above on the graph $\mathcal G$, the
adjacency matrix 
$\D$ is a bounded, self-adjoint operator on the Hilbert space $\ell^2(V)$ with 
respect to the standard scalar product $\langle \phi,\psi\rangle:=\sum_{v\in V} 
\bar{\phi}(v) \psi(v)$ for $\phi,\psi\in\ell^2(V)$. With some abuse of
terminology, $\D$ is also called the (discrete) \textbf{Laplace operator} or 
Laplacian.


The total \textbf{energy}, $H:=-\D+q$, of a quantum mechanical particle on the graph 
$\mathcal G$ is described here by the \textbf{kinetic} energy being equal to 
the negative of the adjacency matrix plus a 
\textbf{potential} energy term given in terms of a bounded function $q:V\to\R$. 
We identify $q$ with the multiplication operator on $\ell^2(V)$ by this 
function $q$ and call $H$ a \textbf{Schr\"odinger operator}.
$H$ is then also a bounded, self-adjoint operator on $\ell^2(V)$. 

$\l\in\C$ is in the resolvent set of $H$, if the so-called \textbf{resolvent}, 
$G_\l:=(H-\l)^{-1}$, of $H$ exists and if $G_\l$ is a bounded operator on 
$\ell^2(V)$. The complement, $\s(H)$, of the resolvent set in $\C$ is called 
the \textbf{spectrum} of $H$. Since $H$ is bounded and self-adjoint, $\s(H)$ 
is a closed, bounded subset of $\R$. 

By the Spectral Theorem (cf.~\cite[Theorem VII.6]{RS}), there exists a family 
of orthogonal projections, $P_\O$, on $\ell^2(V)$ indexed by the 
Borel-measurable sets $\O\subseteq\R$ so that 
\be \label{spectral thm} H = \int_\R t\, dP_t
\ee
with $P_t:=P_{(-\infty,t]} = 1_{(-\infty,t]}(H)$, and $1_\O$ being the 
indicator function of $\O$.

The integral on the right-hand side of \eqref{spectral thm} is meant as a 
Lebesgue-Stieltjes integral so that
$$ \langle\phi, H\psi\rangle = \int_\R t \,d\langle\phi,P_t \psi\rangle\,,
\quad\phi,\psi\in\ell^2(V)\,.
$$
By setting $\mu_{\phi,\psi}(\O):=\langle\phi,P_\O\psi\rangle$ we define a 
(complex) Borel measure, $\mu_{\phi,\psi}$, on $\R$, called a spectral measure
(of $H$). 


Let $\l$ be in the \textbf{upper half-plane} $\H:=\{x+\i y:x,y\in\R,y>0\}$ 
(more generally, $\l$ in the resolvent set of $H$). Then, the kernel of the 
resolvent of $H$ (for $v,w\in V$ we set $1_v := 1_{\{v\}},\mu_{v,w}:=
\mu_{1_v,1_w}$),
\be \label{2.4}
G_{\l}(v,w) = (H-\l)^{-1}(v,w) = \langle 1_v,(H-\l)^{-1} 1_w\rangle
=\int_\R \frac{d\mu_{v,w}(t)}{t-\l}\,,
\ee
is called the \textbf{Green function}; the last identity in \eqref{2.4} follows
from the Spectral Theorem (cf.~\cite[Theorem VII.6]{RS}). In other words, the 
Green function, $G_{\l}(v,w)$, is the Borel transform of the spectral measure, 
$\mu_{v,w}$. Note that (by definition) $G_\l(\cdot,w)$ is the unique function 
$\phi\in\ell^2(V)$ satisfying
\be\label{eq9} (H-\l) \phi = 1_w\,,\quad w\in V\,.
\ee
We Lebesgue-decompose (cf.~\cite[Theorem I.14]{RS}) the probability measure 
$\mu_v:=\mu_{v,v}$ with respect to the Lebesgue measure on $\R$ into its unique
absolutely continuous measure, $\mu_{ac,v}$, and singular measure, $\mu_{s,v}$, 
and write
\be
\mu_v  =  \mu_{ac,v} \oplus \mu_{s,v}\,.
\ee
$\l\in\s(H)$ is said to be in the \textbf{absolutely continuous} (ac henceforth)
or singular spectrum of $H$, if for a vertex $v\in V$, $\l\in\mathrm{supp}(
\mu_{ac,v})$, respectively if $\l\in\mathrm{supp}(\mu_{s,v})$. We are here only 
interested in the ac spectrum of $H$, $\s_{ac}(H)$.

We use a sufficient criterion (see \cite[Theorem 4.1]{Klein},
\cite[Theorem 2.1]{SimonLp}) for 
$\l\in\s(H)$ to be in $\s_{ac}(H)$, namely that there exists an interval 
$(c,d)\ni \l$ and a vertex $v\in V$ so that 
\be\label{eq7} 
\limsup_{\eps\downarrow0} \sup_{\l\in(c,d)} \big|G_{\l+\i\eps}(v,v)\big|
\le C \,,
\ee
for some constant $C$; in fact, $(c,d)\cap \s(H)$ is then in $\s_{ac}(H)$. 
This follows from Stone's formula, which says that for $c,d\in\R, c<d$, and 
for all $\phi\in\ell^2(V)$,
\bea\label{eq6}\lefteqn{\lim_{\eps\downarrow0} \mfr1{\pi} \int_c^d  \Im\, 
\langle \phi,G_{\l+\i\eps}\phi\rangle \, d\l}
\nonumber\\ 
&=&\lim_{\eps\downarrow0} \mfr1{2\pi \i} \int_c^d  \langle\phi,\big[
(H-\l-\i\eps)^{-1}-(H-\l+\i\eps)^{-1}\big]\phi\rangle\, d\l\nonumber\\
&=&\mfr12 \langle \phi, (P_{[c,d]} + P_{(c,d)})\phi\rangle\,.
\eea
Consequently, if $f\in L^q([c,d])$ with $q>1$ and $1/q+1/p=1$, then with
$\phi=1_v$, 
\beax \Big|\int_c^d f(\l)\, d\mu_v(\l)\Big|&\le& \|f\|_q  \,
\limsup_{\eps\downarrow0} \Big(\int_c^d \big[\mfr1\pi\Im (G_{\l+\i\eps}(v,v))
\big]^p\Big)^{1/p}\\
&\le& C \|f\|_q\,.
\eeax
Therefore, by duality, $d\mu_v(\l) = g(\l) d\l$ for some $g\in L^p([c,d])$. 
\bigskip

A \textbf{random potential} is a measurable function $q$ from a measure space
$(A,\mathcal A)$ into $\R^V$, where $\R^V$ is equipped with the Borel product 
$\s$-algebra. In the simplest case there is a single 
probability measure $\nu$ on $\R$ which in turn defines a probability measure, 
$\mathbb P$, on $(A,\mathcal A)$ by requiring the following conditions:
\begin{enumerate}
\item[(i)] $\mathbb P [\omega \in A: q(\omega)(v)\in \O] = \nu(\O)$ for all 
Borel set $\O\subseteq\R$ and for all $v\in V$ ($q$ is said to be identically 
distributed);
\item[(ii)] $\mathbb P \big[\omega \in A : \bigcap_{i=1}^N q(\omega)(v_i)\in 
\O_i)\big]  = \prod_{i=1}^N \mathbb P [\omega \in A : q(\omega)(v_i)\in \O_i]$ 
for all $v_i\not=v_j$ if $i\not= j$, for all $N\in\mathbb N$, and all Borel sets 
$\O_i\subseteq\R$ ($q$ is said to be independently distributed).
\end{enumerate}
We will, without loss of generality, always assume that the mean of $\nu$ is zero
and, to simplify matters, that $\nu$ is compactly supported. The random
Schr\"odinger operator $H:=-\D +q$ on $\ell^2(V)$ with iid random potential 
(that is, $q$ satisfying conditions (i) and (ii) above) is called the 
\textbf{Anderson Hamiltonian} (or \textbf{model}) on the graph $\mathcal G$. 

For $\Im(\l)>0$, the random Green function, $G_\l(v,v)$, (the dependence on
$\omega\in A$ is tacitly suppressed) is a random variable on
$\H$ but simply referred to as the Green function. Since the potential is random so is the spectrum of $H=-\D+q$. However, 
Kirsch and Martinelli~\cite{KM} proved under some (ergodicity) conditions 
on the graph $(V,E)$ --- which are basically\footnote{If we wanted  
ergodicity to be satisfied we should switch from the rooted graphs $\N_0^d$ and 
$T_k$ to $\Z^d$, respectively the unrooted tree. But as much as the ac
spectrum is concerned there is no difference and we stick with the rooted
graphs.} satisfied for $\N_0^d$ and $T_k$ --- 
that the set $\s_{ac}(H)$ is $\mathbb P$-almost surely equal to one 
specific set. Most of the time, the probability measure, $\mathbb P$, is not
mentioned explicitly.

Let $\rho_{\lambda,v}$ be the probability distribution of $G_\l(v,v)$, that is, 
$\rho_{\lambda,v}(A):=\mathrm{Prob}[$ $G_\l(v,v)\in A]$ for a Borel subset 
$A\subseteq\mathbb H$. In order to prove ac spectrum of $H$ we show, loosly 
speaking, that the support of $\rho_{\l,v}$ does not leak 
out to the boundary of $\H$ as $\Im(\l)\downarrow0$ but that the support stays 
inside $\H$. More precisely, for a suitably chosen weight function\footnote{$w$
satisfies $\Im(z)\le C w(z)$ for $z$ near the boundary of $\mathbb H$ with some
constant $C$, 
see \cite[(5)]{FHS3}.} $w$ on $\mathbb H$ (later denoted by $\mathrm{cd}$),
a suitably chosen interval $(c,d)$, and some $p>1$ we shall prove that 
(see \cite[Lemma 1]{FHS2})
\begin{equation}
\limsup_{\eps\downarrow0} \sup_{\l\in(c,d)} \int_{\mathbb H}
w(z)^p d\rho_{\l+\i\eps,v}(z)<\infty\,.
\end{equation}

Let us scale the potential $q$ by the so-called \textbf{disorder}
parameter $a\ge0$ and define $H_a:=-\D+a\,q$. A version of the 
\textbf{extended states conjecture} on a graph $\mathcal G$ can now be formulated as 
the property whether for a probability measure $\nu$ on $\R$ and random 
potential $q$ defined through $\nu$ (obeying the above conditions) and for 
small coupling $a>0$, the ac spectrum of $H_a$ is $\mathbb P$-almost surely 
non-empty, possibly equal to $\s(-\D)$. It is widely believed that this 
conjecture is true on $\N_0^d$ for $d\ge3$ but it is well-known not to be true in 
dimension one. We present a proof of the extended states conjecture on the binary
tree in Section \ref{tree}.
%
%
%
%

\section{One-dimensional graph, $\N_0$}\label{sec:one-dim}

We recall here the standard method of transfer matrices and prove some simple
geometric properties. This is applied to reproving some known results about
decaying potentials. 

Our goal is to bound the diagonal Green function, $G_\l(v,v)$, for $\l\in\H$ 
as $\Im(\l)\downarrow0$ as in \eqref{eq7}. For the sake of simplicity, let us 
take $v=0$. Let $\phi=(\phi_0,\phi_1,\ldots)$ with $\phi_n:=G_\l(0,n)$. 
By recalling (\ref{eq9}), $\phi$ satisfies
\be (-\D+q-\l)\phi = 1_0\,.
\ee
This is equivalent to the system of equations
\bea\label{eq10} -\phi_1 + (q_0-\l) \phi_0 -1&=&0\,,\\
-\phi_{n+1} + (q_n-\l) \phi_n -\phi_{n-1}&=&0\,,\quad n\ge1\,.
\nonumber\eea
Let
\be A_n:=\left[\!\!\begin{array}{cc}q_n-\l&-1\\1&0\end{array}\!\!\right]\,,\quad
n\in\N_0\,.
\ee
$A_n$ is called a \textbf{transfer matrix}. Clearly, $A_n\in \mathrm{SL}(2,\C)$, 
that is, $\det(A_n)=1$. $\phi$ satisfies (\ref{eq10}) if and only if for all $n\ge0$,
\be \label{eq11}
\left[\!\!\begin{array}{c}\phi_{n+1}\\\phi_n\end{array}\!\!\right] = 
A_n A_{n-1} \cdots A_0 \left[\!\!\begin{array}{c}\phi_{0}\\1\end{array}
\!\!\right]\,.
\ee
There is a unique choice of $\phi_0\in \C$, namely $G_\l(0,0)$, so that 
$\phi_n$, computed from (\ref{eq11}), yields a vector $\phi\in\ell^2(\N_0)$. 
An equivalent formulation of (\ref{eq11}) is 
\be \label{eq12}
\left[\!\!\begin{array}{c}\phi_{0}\\1\end{array}\!\!\right] = 
A_0^{-1} A_{1}^{-1} \cdots A_n^{-1} \left[\!\!\begin{array}{c}\phi_{n+1}\\\phi_n
\end{array}\!\!\right]\,.
\ee
Here we compute $\phi_0$ from the likewise unknown vector 
$[\phi_{n+1},\phi_{n}]^T$. Nevertheless,
there is a big difference between (\ref{eq11}) and (\ref{eq12}) when it comes 
to computing $\phi_0$.

As an example let us consider the case without a potential, that is, with
$q=0$. Since $\l\in\H$, the matrix $A_i=\left[\!\!\begin{array}{cc}-\l
&-1\\1&0\end{array}\!\!\right]$ has an eigenvalue $\mu_1$ with $|\mu_1|<1$ and
$\Im(\mu_1)>0$, and another eigenvalue $\mu_2$ with $|\mu_2|=1/|\mu_1|>1$ and 
$\Im(\mu_2)<0$. For $\phi\in\ell^2(\N_0)$ we have to choose $\phi_0$ so that 
$[\phi_{0},1]^T$ is an eigenvector
to $\mu_1$. Therefore, the left-hand side of (\ref{eq11}), namely the vector 
$[\phi_{n+1},\phi_n]^T$ is very sensitive to the choice of the input vector 
$[\phi_{0},1]^T$.
In contrast, the left-hand side of (\ref{eq12}) (for large $n$) is quite
insensitive to the choice of the input vector $[\phi_{n+1},\phi_n]^T$.
Here, the large $n$ behavior is dominated by the large eigenvalue $\mu_2$, and 
$[\phi_{n+1},\phi_n]^T$ must not lie in the eigenspace to the eigenvalue 
$\mu_1$.

\medskip
It is convenient to rewrite the system of equations (\ref{eq12}), and define 
for $\phi = (\phi_n)_{n\in\N_0}$ the sequence $\a=(\a_n)_{n\in\N_0}$ with
\be\label{eq13} \a_{n}:=\frac{\phi_n}{\phi_{n-1}}\,,\quad \phi_{-1}:=1\,.
\ee
Note that for $\l\in\H$, $\phi_n\not=0$: For otherwise, $\l\in\H$ would be an
eigenvalue with eigenfunction $\phi$ of the self-adjoint operator $H$ restricted
to $\{n,n+1,\ldots\}$ (with Dirichlet boundary condition at $n$).

Let $\Phi_n:\H\to\H$ be the \textbf{M\"obius transformation} associated with the 
transfer matrix $A_n^{-1}$. That is,
\be\label{eq14}
\Phi_n(z):=-\frac1{z+\l-q_n}\,.
\ee
Then (\ref{eq12}) is equivalent to 
\be\label{eq15}
\phi_0 = \Phi_0\circ\Phi_1\circ\cdots\circ \Phi_{n}(\a_{n+1})\,\,.
\ee
The numbers $\a_n$ can be interpreted as the Green function of the graph $\N_0$
truncated at $n$. To this end, let $\N_n:=\{n,n+1,\dots\}$ and
$E_n:=\{(k,k+1),(k+1,k),k\ge n\}$. If $\D_n^{(t)}$ denotes the adjacency matrix 
for the truncated graph $(\N_n,E_n)$, then
\be \label{def:alpha}
\a_n = G_\l^{(t)}(n,n):=(-\D_n^{(t)}+q-\l)^{-1}(n,n)\,,\quad n\in\N_0\,,
\ee
and we have the recursion
\be \label{recursion}
\a_n = \Phi_n(\a_{n+1})\,,\quad n\in\N_0\,.
\ee
This can be seen from the above equations but we will rederive this later, see 
formula (\ref{eq22}). 


We equip the upper half-plane $\H$ with the hyperbolic
(or Poincar\'e) metric $\mathrm{d}$, that is,
\be \label{def:metric}
\mathrm{d}(z_1,z_2):=\cosh^{-1}\Big(1+\mfr12\,\frac{|z_1-z_2|^2}{\Im(z_1) 
\Im(z_2)}\Big)\,,\quad z_1,z_2\in\H\,,
\ee
or alternatively with the Riemannian line element (see also \eqref{4.10} and
\eqref{def:Finsler}),
\be\label{dsRiem} ds = \frac{\sqrt{dx^2+dy^2}}{y}\,,\quad z=x+\i y\in\H\,.
\ee

\begin{prop}[\!\!\cite{FHS1}, Proposition 2.1]\label{prop1}  \begin{enumerate}
\item[(i)] For $\Im(\l)\ge0$, $\Phi_n$ is a hyperbolic contraction on 
$(\H,\mathrm{d})$, that is, for $z_1,z_2\in\H$,
$$ \mathrm{d}(\Phi_n(z_1),\Phi_n(z_2)) \le \mathrm{d}(z_1,z_2)\,.
$$
\item[(ii)] For $\Im(\l)>0$, $\Phi_n(\H)\subset \{z\in\H:|z|<1/\Im(\l)\}$. 
Furthermore, $\Phi_n$ is a strict hyperbolic
contraction. That is, for $z_1,z_2\in\H$ with
$\max\{|z_1|,|z_2|\}<C$ there exists a constant $\d<1$, e.g.,
$\d:=C/(C+\Im(\l))$, depending on $\Im(\l)$ and $C$ so that
$$ \mathrm{d}(\Phi_n(z_1),\Phi_n(z_2)) \le \d \,\mathrm{d}(z_1,z_2)\,.
$$
\end{enumerate}
\end{prop}
The basic idea is to factor $\Phi_n = \rho\circ\tau_n$ into the rotation 
(around the point $\i$ and angle $\pi$) $\rho:z\mapsto -1/z$ and the translation
$\tau_n:z\mapsto z+\l-q_n$. $\rho$ is a hyperbolic isometry. If $\Im(\l)>0$ then
$\tau_n$ is a strict hyperbolic contraction in the sense that $\mathrm{d}(\tau_n
(z_1),\tau_n(z_2)) < \mathrm{d}(z_1,z_2)$ as can be seen directly from definition 
\eqref{def:metric}. If $\Im(\l)=0$, then also $\tau_n$ is an isometry. The 
properties claimed in (i) and (ii) follow from straightforward calculations.  
$\quad\Box$

If $\Im(\l)>0$ then $\Phi_n$ shifts the upper half-plane upwards. Even more 
so (recall that the potential $q$ is bounded) we have
\begin{prop}[\!\!\cite{FHS1}, Proposition 2.2]\label{prop2} For $\Im(\l)>0$ 
there exists a hyperbolic disk $B\subset\H$ so that $\Phi_{n-1}\circ\Phi_{n}
(\H)\subset B$.
\end{prop}
This allows us to state precisely our claim about the stability of our way to 
compute the Green function.
\begin{thm}[\!\!\cite{FHS1}, Theorem 2.3]\label{thm3} Let $\Im(\l)>0$ and let 
$(\g_n)_{n\in\N}$ be an \emph{arbitrary} sequence in $\H$. Then we have
\be
\lim_{n\to\infty} \Phi_0\circ\Phi_1\circ\cdots\circ \Phi_n(\g_n) = \phi_0 =
G_\l(0,0)\,. 
\ee
\end{thm}

\begin{proof} Set $w_n:=\Phi_0\circ\Phi_1\circ\cdots\circ \Phi_n(\g_n)$. Let
$B$ be a disk as in Proposition \ref{prop2}, and let $\b:=\Phi_{n-1}\circ
\Phi_n(\g_n)$. Then $\b\in B$. The same is true for $\b':=\Phi_{n-1}\circ
\Phi_n\circ\Phi_{n+1}(\g_{n+1})$. All further images $\Phi_k(\b)$ and 
$\Phi_k(\b')$ stay in $B$ and the conditions from Proposition \ref{prop1}(ii) 
are fulfilled. Hence we have 
\beax \mathrm{d}(w_{n+1},w_n)
&=&\mathrm{d}(\Phi_0\circ\cdots\circ\Phi_{n-2}(\b'),\Phi_0\circ\cdots\circ
\Phi_{n-2}(\b))
\\
&\le&\d\,\mathrm{d}(\Phi_1\circ\cdots\circ\Phi_{n-2}(\b'),\Phi_1\circ\cdots\circ
\Phi_{n-2}(\b))
\\
&\le&\d^{n-1}\mathrm{d}(\b',\b)
\\
&=&C\,\d^n\,.
\eeax
$(w_n)_{n\in\N}$ is therefore a Cauchy sequence and converges to some 
$w\in\H$. Let $(\g'_n)_{n\in\N}$ be another sequence in $\H$. Then we have
analoguously
\beax\lefteqn{\mathrm{d}(\Phi_0\circ\cdots\circ\Phi_{n-1}\circ\Phi_n(\g_n),
\Phi_0\circ\cdots\circ\Phi_{n-1}\circ\Phi_n(\g'_n))}\hspace{1cm}
\\
&\le& C \,\d^{n-1} \mathrm{d}(\Phi_{n-1}\circ\Phi_n(\g_n),\Phi_{n-1}\circ
\Phi_n(\g'_n))
\\
&\le& C\,\d^{n-1}\,.
\eeax
Therefore also $\Phi_0\circ\cdots\circ\Phi_{n-1}\circ\Phi_n(\g'_n)$ converges to
$w$ as $n\to\infty$. Because of (\ref{eq15}), $w=\phi_0=G_\l(0,0)$. 
\end{proof}

\begin{prop}[\!\!\cite{FHS1}, Lemma 4.5]\label{prop4} Let $K$ be a compact 
subset of $\C$ whose elements have non-negative imaginary parts. For every 
$\l\in K$, let $(z_n(\l))_{n\in\N}$ be a sequence in $\H$. Suppose that there 
exist constants $C_1,C_2$ so that
\be\label{eq16}
\sum_{n\ge1}\mathrm{d}\big(\Phi_{n+1}(z_{n+1}(\l)),z_n(\l)\big)\le C_1 
\ee
and
\be\label{eq17}
\mathrm{d}\big(z_1(\l),\i\big) \le C_2
\ee
for all $\l\in K$. Then there exists a constant $C_3$ so that for all
$\l\in K$
\be\label{eq18}
\mathrm{d}\big(G_\l(0,0),\i\big) \le C_3\,.
\ee
\end{prop}
Potentials for which we can find such sequences $(z_n(\l))_{n\in\N}$ yield ac 
spectrum for $\l\in \Re(K)$, and pure ac spectrum for $\l\in\mathrm{int}
(\Re(K))$, the interior of the real part of $K$.
\begin{proof} Because of Theorem \ref{thm3} there exists an $n\in\N$ so that
$\mathrm{d}\big(G_\l(0,0),\Phi_0\circ\cdots\circ\Phi_{n}(z_n)\big)\le 1$.
Then using the triangle inequality for the Poincar\'e metric $\mathrm{d}$ 
and the contraction property of $\Phi_n$ we get (suppressing the dependence of
$z_n$ on $\l$),
\beax\mathrm{d}\big(G_\l(0,0),\i\big)
&\le&\mathrm{d}\big(G_\l(0,0),\Phi_0\circ\cdots\circ\Phi_{n}(z_n)\big) + 
\mathrm{d}\big(\Phi_0\circ\cdots\circ\Phi_{n}(z_n),\i\big)
\\
&\le&\mathrm{d}\big(\Phi_0\circ\cdots\circ\Phi_{n-1}(\Phi_{n}(z_n)),
\Phi_0\circ\cdots\circ\Phi_{n-1}(z_{n-1})\big) 
\\
&& + \;\mathrm{d}\big(\Phi_0\circ\cdots\circ\Phi_{n-1}(z_{n-1}),\i\big) + 1
\\
&\le&\mathrm{d}\big(\Phi_{n}(z_{n}),z_{n-1}\big) + 
\mathrm{d}\big(\Phi_0\circ\cdots\circ\Phi_{n-1}(z_{n-1}),\i\big) + 1
\\
&\le&\ldots
\\
&\le&\sum_{k=1}^{n-1} \mathrm{d}\big(\Phi_{k+1}(z_{k+1}),z_{k}\big)
+ \mathrm{d}\big(\Phi_0(z_1),\i\big) + 1
\\
&\le&C_1 + \mathrm{d}\big(\Phi_0(z_1),\i\big) + 1:=C_3\,.
\eeax
\vskip-2em
\end{proof}

\begin{exs} \begin{enumerate}
\item[(i)] Zero potential: Here, $\Phi_n(z) = -\frac1{z+\l}$. For $\l\in\H$, 
let $z_{+}(\l)\in\H$ be the fixed point of $\Phi_n$, that is,
\be\label{fpglg}
z_{+}(\l) = -\frac1{z_{+}(\l)+\l}\,.
\ee
Using Theorem \ref{thm3} with $\g_n=z_{+}(\l)$ we get $\phi_0=G_\l(0,0)=
z_{+}(\l)$. We have
\be z_{+}(\l) = -\l/2+ \i\sqrt{1-\l^2/4} \,.
\ee
$z_{-}(\l):=-\l/2- \i\sqrt{1-\l^2/4}$ is the second solution to the fixed point
equation (\ref{fpglg}), but it lies in the lower half-plane. $z_{\pm}(\l)$ are 
also the two eigenvalues of the transfer matrix. $z_+(\l)$ and $z_-(\l)$ are 
the stable respectively unstable eigenvalue of this matrix. For $\l\in\R$, 
$z_+(\l)\in\H$ if and only if $|\l|<2$. Therefore, $\s(-\D)=\s_{ac}(-\D)=
[-2,2]$. 


\item[(ii)] Short-range potential $q$, that is, $\sum_n |q_n|<\infty$: We 
choose the constant sequence $(z_n)_{n\in\N}$ with $z_n:=z_{+}(\l)$ for 
$n\in\N$. Then we have
\beax \mathrm{d}(\Phi_{n}(z_{n}),z_{n})
&=&\mathrm{d}\Big(z_{n}+\l-q_n,-\frac1{z_{n}}\Big)
\\
&=&\mathrm{d}\big(\l/2+ \i\sqrt{1-\l^2/4}-q_n,\l/2+ \i\sqrt{1-\l^2/4}\big)
\\
&\le&C \,|q_n|\,.
\eeax
By Proposition \ref{prop4}, $[-2,2]\subseteq \s_{ac}(-\D + q)$, and on 
$(-2,2)$ the spectrum is purely ac. 

\item[(iii)] A Mourre estimate: Suppose that $\sum_{n\ge1} |q_{n+1}-q_n| <
\infty$. Choose now $z_n$ for $n\ge k$ to be the fixed point of the map 
$\Phi_n$. Then $z_n = -(\l-q_n)/2 + \i \sqrt{1-(\l-q_n)^2/4}$. $z_n\in\H$ if 
$|\l-q_\infty|<2$ with $q_\infty:=\lim_{n\to\infty} q_n$ and $k$ large enough. 
For $1\le n< k$ choose arbitrary points in $\H$. Then we have
$$\mathrm{d}(\Phi_{n+1}(z_{n+1}),z_n)=\mathrm{d}(z_{n+1},z_n) \le C \,
|q_{n+1}-q_n|\,.
$$
By Proposition \ref{prop4}, $[-2+q_\infty,2+q_\infty]\subseteq\s_{ac}
(-\D + q)$. Note, for instance, that by this Mourre estimate, a monoton
potential decaying to zero always has pure ac spectrum inside $(-2,2)$.

\end{enumerate}
\end{exs}

By allowing the potential to be random, the decay conditions on the potential
can be weakened to guarantee ac spectrum. In one dimension, the
$\ell^1$-condition can then be replaced by an $\ell^2$-condition.

\begin{thm}[\!\!\cite{FHS4}, Theorem 1]\label{thm5} 
Let $q=(q_n)_{n\in\N_0}$ be a family of centered, independent, real-valued 
random variables with corresponding probability measures $\nu_n$ and all with support in 
some compact set $K$. Suppose that $\E[\sum_{n\ge0}|q_n|^2]< \infty$, where 
$\E$ is the expectation with respect to the product measure, 
$\nu=\bigotimes_{n\ge0}\nu_n$. Then almost surely, $[-2,2]$ is part of the ac 
spectrum of $H$, and $H$ is purely ac on $(-2,2)$.
\end{thm}

\begin{rems} \begin{enumerate}
\item[(i)] Deylon-Simon-Souillard~\cite{DSS85} have proved Theorem \ref{thm5} 
in 1985 even without assuming compact support of the probability measure.
Furthermore, they proved that if $C^{-1}n^\rho \le \E[|q_n|^2]\le C n^\rho$ 
for some constant $C$ and $\rho<1/2$, then 
the spectrum of $H=-\D+q$ is pure point (almost surely) with exponentially
localized eigenfunctions.
\item[(ii)] In \cite{FHS4}, we have extended Theorem \ref{thm5} to matrix-valued 
potentials, and applied to (random) Schr\"odinger operators on a strip.
\item[(iii)] On the two-dimensional lattice $\N_0^2$, Bourgain~\cite{B02}
proved $\sigma_{\rm ac}(\Delta+q) = \sigma(\Delta)$ for centered 
Bernoulli and Gaussian distributed, independent random potentials 
with $\sup_{n\in\N_0^2} \E[q_n^2]^{1/2}\, |n|^\rho <\infty$ for $\rho>1/2$.  
In~\cite{B03}, Bourgain improves this result to $\rho>1/3$.
\end{enumerate}
\end{rems}

\begin{proof}[Proof of Theorem \ref{thm5}] For $\l\in (-2,2)$ let 
$z_\l:=-\l/2+\i\sqrt{1-\l^2/4}$ be the (truncated) Green function of the 
Laplace operator $-\D$, see \eqref{fpglg}. Let us introduce the weight function
\be \label{def:cd1}
\mathrm{cd}:\H\to(0,\infty)\,,z\mapsto \frac{|z-z_\l|^2}{\Im(z)}\,.
\ee
By Proposition \ref{prop2} there is a disk $B\subset \H$ so that
$z_{0,n}:=\Phi_0\circ\cdots\circ\Phi_n(z_\l)\in B$ for all $n\ge2$ and potentials
$q$ with values in a compact set $K$. Moreover, by Theorem \ref{thm3},
$G_\l(0,0) = \lim_{n\to\infty} z_{0,n}$. Hence, by the continuity of the
function $\mathrm{cd}$, we have $\lim_{n\to\infty}\mathrm{cd}^2(z_{0,n}) =  
\mathrm{cd}^2(G_\l(0,0))$. Since $\mathrm{cd}^2$ is bounded on the disk $B$ 
we conclude that $\mathbb E(\mathrm{cd}^2(G_\l(0,0))) = \lim_{n\to\infty} \mathbb 
E(\mathrm{cd}^2(z_n))$. It remains to show that this limit is bounded. To this end, 
we define the rate of expansion, 
\be\mu(z,q_n):=\frac{\mathrm{cd}^2(\Phi_n(z))+1}{\mathrm{cd}^2(z)+1}\,.
\ee
Noticing that $\mathrm{cd}(\Phi_n(z)) = |z-z_\l-q_n|^2/\Im(z+\l)$ and using 
$|q_n|\le C$ to bound cubic and quartic terms of $q$ in terms of quadratic 
ones, we obtain that
\be\mu(z,q_n)\le A_0(z) + A_1(z) q_n + A_2(z) q_n^2
\ee
with rational functions $A_i(z)$. The functions $A_1$ und $A_2$ are bounded and
$A_0\le1$. Let us set $z_{\ell,n} := \Phi_{\ell} \circ \Phi_{\ell+1} \circ
\cdots \circ \Phi_{n}(z_\lambda)$. Note that $z_{\ell,n} = \Phi_{\ell}
(z_{1,n})$. 
By the recursion relation \eqref{recursion}, 
\begin{eqnarray*}\lefteqn{
\E[{\rm cd}^2(z_{0,n})] + 1}
\\
&=& \int_{K^{n+1}} \big({\rm cd}^2(z_{1,n}) + 1 \big) \, d \nu_0(q_0)
\cdots d \nu_n(q_n)
\\
&=& \int_{K^{n+1}} \frac{{\rm cd}^2[\Phi_{0} (z_{1,n})] + 1 }{{\rm
cd}^2(z_{1,n}) + 1} \,\big({\rm cd}^2(z_{1,n}) + 1 \big) \,
d \nu_0(q_0) \cdots d \nu_n(q_n)
\\
&\leq& \int_K \big( 1 + A_1(z_{1,n}) q_0  + C_0 q_0^2 \big)\, d\nu_0(q_0)
\int_{K^n} \big({{\rm cd}^2(z_{1,n}) + 1}\big) \,
d \nu_1(q_1) \cdots d \nu_n(q_n)
\\
&=& \big( 1 + C_0 \,\E[q_0^2] \big) \int_{K^n} \big({\rm cd}^2(z_{1,n}) + 1\big) \,
d \nu_1(q_2) \cdots d \nu_n(q_n)\\
&\leq& \prod_{i=0}^n \big( 1 + C_0 \,\E[q_i^2] \big)\;\leq\;  
\exp  \big(C_0 \sum_{i=0}^\infty \E[q_i^2] \big) \;<\infty\,.
\end{eqnarray*}
%
%
%
%
%
\vskip-3em
\end{proof}

\section{General graphs}\label{higher dimension}

We generalize the approach of the previous section to calculating the Green
function via transfer matrices (or rather M\"obius transformations) to  
general graphs $\mathcal G=(V,E)$, that is, to all graphs that obey 
the conditions 
(i) and (ii) of Section \ref{setup}. Let us choose a point in $V$ 
which we denote by $0$. If $\mathrm{dist}(v,w)$ is the graphical distance 
between the two lattice points $v$ and $w$ then we define the $n$-th sphere,
\be\label{eq19}
S_n:=\{v\in V:\mathrm{dist}(v,0) = n\}\,.
\ee
Clearly, $V=\bigcup_{n\ge0} S_n$ and $\ell^2(V) = \bigoplus_{n\ge0} 
\ell^2(S_n)$. We decompose the adjacency matrix, $\D$, of $\mathcal G$ into the 
block matrix form
\be\label{eq20}
\D = \left[\!\!\begin{array}{cccccc}D_0&E_0^T&0&\cdots&\cdots&\cdots\\
E_0&D_1&E_1^T&0&\cdots&\cdots\\0&E_1&D_2&E_2^T&0&\cdots\\
\vdots&\ddots&\ddots&\ddots&\ddots&\ddots\end{array}\!\!\right]\,,
\ee
where $D_n$ is the adjacency matrix of $\mathcal G$ restricted to
$S_n$. $E_n:\ell^2(S_n)\to\ell^2(S_{n+1})$ is the map with kernel
$$ E_n(v,w) = \left\{\begin{array}{ccc}1&\mbox{ if }v\in S_n,w\in S_{n+1},
(v,w)\in E\\0&\mbox{ else }\end{array}\right.\,.
$$
The potential $q = \bigoplus_{n\ge0}q_n$ is diagonal; $q_n$ equals the
restriction of $q$ to the sphere $S_n$ which is now considered a 
$|S_n|$-dimensional diagonal matrix. $H= -\D +q$ is then of the block matrix 
form
\be\label{eq21}
H = \left[\!\!\begin{array}{cccccc}-D_0+q_0&-E_0^T&0&\cdots&\cdots&
\cdots\\-E_0&-D_1+q_1&-E_1^T&0&\cdots&\cdots\\0&-E_1&-D_2+q_2&-E_2^T&0&\cdots\\
\vdots&\ddots&\ddots&\ddots&\ddots&\ddots\end{array}\!\!\right]\,.
\ee
Let $P_n:\ell^2(V)\to\ell^2(S_n)$ be the orthogonal projection of $\ell^2(V)$
onto $\ell^2(S_n)$, and $P_{n,\infty}:=\sum_{k\ge n}P_k$. Then we define the 
truncated Hamiltonian
\be H_n:=P_{n,\infty} \,H\, P_{n,\infty}
\ee
and the truncated Green function
\be\label{eq22}
G_\l^{(t)}(n,n):=P_n(H_n-\l)^{-1}P_n\,,\quad n\in\N_0\,.
\ee
$G_\l^{(t)}(n,n)$ is a $d_n\times d_n$ dimensional matrix with $d_n=|S_n|$. By 
definition, $G_\l^{(t)}(0,0)$ equals the Green function, $G_\l(0,0)$.
Furthermore (assuming as usual $\l\in\H$),
\be G_\l^{(t)}(n,n) \in {\mathbb S\mathbb H}_{d_n}\,,
\ee
where 
$$ {\mathbb S\mathbb H}_{d} :=\{Z=X + \i Y: X,Y\in \mathrm{Mat}(d,\R), 
X=X^T,Y>0\}
$$
is the so-called \textbf{Siegel half-space}. Clearly, ${\mathbb S\mathbb H}_1 = \H$. 

The matrices $G_\l^{(t)}(n,n)$ generalize the numbers $\a_n\in\H$ from
(\ref{eq13}). 
More precisely, let $\Phi_n: {\mathbb S\mathbb H}_{d_{n+1}}\times 
\mathrm{Mat}(d_n,\R)\times\H\to{\mathbb S\mathbb H}_{d_n}$ be defined as
\be \Phi_n(Z,q_n,\l):= -(E_n^T Z E_n + D_n -q_n+\l)^{-1}\,.
\ee
Then in analogy with (\ref{def:alpha}) we have 
\be\label{eq23}
G_\l^{(t)}(n,n) = \Phi_n\big(G_\l^{(t)}(n+1,n+1),q_n,\l\big)\,.
\ee
The proof is simply based upon Schur's (or Feschbach's) formula
\be
\left[\!\!\begin{array}{cc}A & B^T\\B & C\end{array}\!\!\right]^{-1} = 
\left[\!\!\begin{array}{cc}(A-B^TC^{-1}B)^{-1} &(B^TC^{-1}B-A)^{-1}B^TC^{-1}\\
C^{-1}B(B^TC^{-1}B-A)^{-1} & (C-BA^{-1}B^T)^{-1}\end{array}\!\!\right]\,,
\ee
by setting $A := -D_n+q_n-\l, B := (E_n,0,\ldots)$ and $C := H_{n+1}-\l$. 

On ${\mathbb S\mathbb H}_{n}$, we do not use the standard Riemann metric but a 
so-called Finsler metric. To this end, let $W\in \mathrm{Mat}(n,\C)$ be an 
element of the tangent space at $Z=X+\i Y\in{\mathbb S\mathbb H}_{n}$. Then we 
set
\be \label{4.10}F_Z(W):=\|Y^{-1/2} W Y^{-1/2}\|\,,
\ee
where $\|\cdot\|$ is the operator norm (rather than the Hilbert-Schmidt norm).
[If $n=1$ then the length of the tangent vector is $|W|/Y$ as in
(\ref{dsRiem}).] The Finsler metric on ${\mathbb S\mathbb H}_{n}$ is defined 
as (thereby suppressing the dimension $n$)
\be \label{def:Finsler}
\mathrm{d}(Z_1,Z_2):=\inf_{Z(t)}\int_0^1 F_{Z(t)}\big(\dot{Z}(t)\big)\,dt\,,
\quad Z_1,Z_2\in {\mathbb S\mathbb H}_{n}\,,
\ee
whereby $Z(t)$ runs through all differentiable paths $Z:[0,1]\to
{\mathbb S\mathbb H}_{n}$ with $Z(0)=Z_1,Z(1)=Z_2$.

The Propositions \ref{prop1}, \ref{prop2}, and \ref{prop4} can be extended to
general graphs, see \cite[Proposition 3.3, Lemma 3.5, Lemma 4.5]{FHS1}. 
For instance, for a fixed potential $q$ and fixed $\l\in\H$, the transformation 
$\Phi_n$ is a contraction from $({\mathbb S\mathbb H}_{d_{n+1}},\mathrm{d})$ 
into $({\mathbb S\mathbb H}_{d_{n}},\mathrm{d})$. Theorem \ref{thm3} generalizes
as follows.
\begin{thm}[\!\!\cite{FHS1}, Theorem 3.6] \label{thm6} Let us assume that the 
matrices $E_i$ in \eqref{eq20} all have kernel $\{0\}$. Let $\Im(\l)>0$ and let
$Z_i\in{\mathbb S\mathbb H}_{d_i}$ with $d_i = |S_{i}|$ be an \emph{arbitrary} 
sequence. Then for a bounded potential $q$ we have
\be \lim_{n\to\infty} \Phi_0\circ\cdots\circ\Phi_n(Z_{n+1},q_n,\l) = \phi_0 = 
G_\l(0,0)\,.
\ee 
\end{thm} 

\section{Trees}\label{tree}

Let us consider for simplicity the (rooted) binary tree, $T_2$. The recursion 
relation \eqref{eq23} is very simple since diagonal matrices are mapped into 
diagonal matrices. Hence, the truncated Green functions (or rather matrices)
are diagonal by Theorem \ref{thm6}. Let $q_n =\mathrm{diag}(q_{n,1},\ldots,
q_{n,2^n})$ be the diagonal matrix with diagonal real-valued entries $q_{n,1},
\ldots,q_{n,2^n}$ and $Z=\mathrm{diag}(z_1,\ldots,z_{2^{n+1}})$ a diagonal 
matrix in ${\mathbb S\mathbb H}_{2^{n+1}}$, that is, with $z_i\in\H$. Then
\be \label{def:Psi}
\Phi_n\big(Z,q_n,\l\big) = \mathrm{diag}\big(\Psi(z_1,z_2,q_{n,1},\l),
\Psi(z_3,z_4,q_{n,2},\l),\ldots\big)
\ee
with the map $\Psi:\H^2\times\R\times\H\to\H$ defined as 
$\Psi(z_1,z_2,q,\l):=-1/(z_1+z_2+\l-q)$. 

Now put $q=0$ and consider the fixed point equation
\be \Psi(z,z,0,\l) = -\frac{1}{2z+\l} = z\,.
\ee
The two solutions are obviously $-\l/4 \pm \sqrt{\l^2/16 - 1/2}$. 
For $\l\in\R$, they have non-zero imaginary component if and only if
$|\l|<2\sqrt{2}$. We choose 
\be \label{fixpnt} z_\l := -\l/4 +\i \sqrt{1/2 - \l^2/16}
\ee
for the solution in $\H$. Furthermore, for $n\in\N_0$ let $Z_n:=
\mathrm{diag}(z_\l,\ldots,z_\l)\in{\mathbb S\mathbb H}_{2^n}$. Then 
\beax \Phi_n\big(Z_{n+1},0,\l\big)=\mathrm{diag}
\big(\Psi(z_\l,z_\l,0,\l),\ldots\big)=Z_{n}\,,\quad n\in\N_0\,.
\eeax 
Theorem \ref{thm6} then shows that $G_\l(0,0)=z_\l$ for the rooted binary tree. 
Hence, 
$$ \s(\D)= \s_{ac}(\D) = [-2\sqrt{2},2\sqrt{2}]\,.
$$

Let us consider the Anderson model, $H_a=-\D+a\, q$, on this tree with iid 
random potential $q$, which is determined by a probability measure, $\nu$. 
For simplicity, we assume that $\nu$ has compact support. 
\begin{thm}[\!\!\cite{FHS2}, Theorem 1] \label{ext states conj} For every $|\l|< 2\sqrt{2}$ there is 
an $a_0>0$ so that for all $0\le a\le a_0$ almost surely
\be \s_{ac}(H_a)\cap (-\l,\l) = (-\l,\l)\,,\quad\s_{s}(H_a)\cap (-\l,\l) = 
\emptyset\,.
\ee
\end{thm}
This has been proved first by Klein~\cite{Klein} in 1998. The statement 
$\s_{ac}(H_a)\cap (-\l,\l) \not=\emptyset$ has been proved by Aizenman, Sims, 
and Warzel~\cite{ASW1,ASW2} in 2005. The following proof is shorter than our first
one presented in~\cite{FHS2} since we now work directly with the Green function
instead of the sum of Green functions, which simplifies the analysis of the
functions $\mu_{2,p}$ and $\mu_{3,p}$ (see the following proof) considerably.

\begin{proof}[Sketch of Proof] Let $\l\in\H$ with $|\Re(\l)|<2\sqrt{2}$.
The truncated Green function,
$G_{\l,a}^{(t)}(n,n)$, is an ${\mathbb S\mathbb H}_{2^n}$-valued random 
variable with range inside the diagonal matrices. In fact, its probability 
distribution equals $\bigotimes_{i=1}^{2^n} \rho_a$, where, for short, $\rho_a$ 
is the probability distribution of $G_{\l,a}(0,0):=(H_a-\l)^{-1}(0,0)$. 
Using the recursion relation (\ref{eq23}) we see that $\rho_a$ equals the 
image measure $(\rho_a\times \rho_a\times \nu_a)\circ\Psi^{-1}$ with the 
function $\Psi$ as in \eqref{def:Psi}.

Now we define a moment of $\rho_a$ that we need to control as $\Im(\l)
\downarrow0$, that is, we are seeking a uniform bound of $M_p(\rho_a)$ below 
as $\Im(\l)\downarrow0$. As in \eqref{def:cd1} but with $z_\l$ from 
\eqref{fixpnt}, let us introduce the weight function
\be \label{def:cd2}
\mathrm{cd}(z) := \frac{|z-z_\l|^2}{\Im(z)}\,,\quad z\in\H\,.
\ee
Then we define for some $p>1$
\be M_p(\rho_a) := \int_\H \mathrm{cd}^p(z)\, d\rho_a(z)\,.
\ee
Applying the recursion relation we get
\begin{eqnarray} \label{mu2p}
\lefteqn{M_p(\rho_a)}\\
&=&\int_{\H^2\times \R}\mathrm{cd}^p\big(\Psi(z_1,z_2,q,\l)\big)\, 
d\rho_a(z_1) d\rho_a(z_2) d\nu_a(q)\nonumber
\\
&=&\int_{\H^2\times \R}\underbrace{\frac{\mathrm{cd}^p\big(\Psi(z_1,z_2,q,\l)
\big)}{\mfr12\mathrm{cd}^p(z_1) + \mfr12\mathrm{cd}^p(z_2)}}_{=:\mu_{2,p}
(z_1,z_2,q,\l)} \,
\big(\mfr12 \mathrm{cd}^p(z_1) + \mfr12\mathrm{cd}^p(z_2)\big)\, d\rho_a(z_1) 
d\rho_a(z_2) d\nu_a(q) \,.\nonumber
\end{eqnarray}
For $z_i\in\H$ and $y_i:=\Im(z_i)$, let $u_i:=(z_i-z_\l)/\sqrt{y_i}\in\C$. Then
$\mathrm{cd}(z_i)=|u_i|^2$. Using $u:=(u_1,u_2)$ and $v:=\big(\sqrt{y_1/
(y_1+y_2)},\sqrt{y_2/(y_1+y_2)}\big)$ we obtain
$$ \mathrm{cd}(\Psi(z_1,z_2,0,\l)) =\mfr12\,\frac{|z_1+z_2-2z_\l|^2}{y_1+
y_2+\Im(\l)} < \mfr12\,\frac{|z_1-z_\l +z_2-z_\l|^2}{y_1+y_2} =  \mfr12\,
\big|\langle u,v\rangle\big|^2\,.
$$
By the Cauchy-Schwarz inequality and the strict convexity of $x\mapsto x^p$ 
for $p>1$ we see that
\be \mu_{2,p}(z_1,z_2,0,\l) \le \frac{\Big(\mfr12\,\big|\langle
u,v\rangle\big|^2\Big)^p}{\mfr12\,|u_1|^{2p}+\mfr12\,|u_2|^{2p}} \le1
\,.
\ee
For $\Im(\l)=0$, the function $\mu_{2,p}(z_1,z_2,0,\l)=1$ if and only if $u=sv$
for some $s\in\C$ and if 
$|u_1|=|u_2|$. The function $(z_1,z_2,\l)\mapsto\mu_{2,p}(z_1,z_2,0,\l)$ is
continuous on $\C^2\times(-2\sqrt{2},2\sqrt{2})$ except at $u=0$. 
This implies that in order to have equality, $z_1$ and $z_2$ have to be of the 
form $z_1=x_1+\i y,z_2=x_2+\i y$ with $|x_1+\l/4| = |x_2+\l/4|$. Obviously, we 
cannot expect that $\mu_{2,p}(z_1,z_2,q,\l)\le1-\mu<1$ for $z_1,z_2$ in a 
neighborhood of the boundary of $\H^2$ with a constant $\mu$ und $q$ in an 
interval $I\ni0$. For the sake of the argument, let us suppose that the 
\emph{average}
$\bar{\mu}_{2,p}(z_1,z_2,\l):=\int_\R \mu_{2,p}(z_1,z_2,q,\l)\,d\nu_a(q)\le1-
\mu<1$ for $z_1,z_2$ in a neighborhood of the boundary of $\H^2$ for small 
enough disorder, $a$.  
Then choose some compact set $B\subset\H^2$ with $B\ni(z_\l,z_\l)$, and split 
the integration into an integral over $B$ and its complement in $\H^2$. On the 
first set, the integrand is bounded and on the second set we use the contraction
property of $\bar{\mu}_{2,p}$. That is,
\begin{eqnarray}\lefteqn{M_p(\rho_a)}\nonumber
\\&=& \int_{(B\times \R)\cup 
(\H^2\setminus B)\times \R} \!\!\mu_{2,p}(z_1,z_2,q,\l) 
\big[\mfr12 \mathrm{cd}^p(z_1) + \mfr12\mathrm{cd}^p(z_2)\big]
   \, d\rho_a(z_1) d\rho_a(z_2) d\nu_a(q)\nonumber
\\
&\le&C + (1-\mu) \int_{\H^2\setminus B} 
      \big[\mfr12 \mathrm{cd}^p(z_1) + \mfr12\mathrm{cd}^p(z_2)\big]\,
      d\rho_a(z_1) d\rho_a(z_2)\nonumber
\\
&=&C+(1-\mu) M_p(\rho_a)\,,
\end{eqnarray}
where $C$ is a finite constant. This implies $M_p(\rho_a)<C/\mu<\infty$. Our
assumption that $\bar{\mu}_{2,p}(z_1,z_2,\l)\le1-\mu<1$ is not quite true. But this
averaging was essential in a similar situation in the proof of Theorem 
\ref{twoperiodic}, see \cite{FHS3}.

In order to obtain an estimate of a corresponding function 
$\mu_{3,p}(\boldsymbol z,\boldsymbol q,\l)\le1-\mu<1$ for $\boldsymbol z = 
(z_1,z_2,z_3)$ in a neighborhood of the boundary of $\H^3$, for $\boldsymbol q=
(q_1,q_2)$ in a small square with center at $\boldsymbol 0$, and for all 
$\l\in (-2\sqrt{2},2\sqrt{2})$ we use the recursion relation one more time.
Before we define this function $\mu_{3,p}$ we extend
$\mu_{2,p}$ to an upper semi-continuous function onto the boundary of $\H^2$ 
(in terms of the $z$ variables) via a radial compactification of $\C^2$
(in terms of the $u$ variables). To this end, let $r>0,(\omega_1,\omega_2)
\in\C^2$ so that
\begin{eqnarray}\label{blowup}
\frac{1}{u_1} = r\omega_1\,,\quad \frac{1}{u_1} = r\omega_1\,,\quad
|\omega_1|^2 + |\omega_2|^2 =1\,.
\end{eqnarray}
Then,
\begin{eqnarray}
\lefteqn{\mu_{2,p}(z_1,z_2,q,\l) =}\\
&&\frac{\Big(\mfr12 |\langle(\omega_2,\omega_1),v\rangle|^2
-q\,r\,\Re\langle(\omega_2,\omega_1),v\rangle +
\mfr12\frac{q^2|\omega_1\omega_2|^2}{
|\omega_1^2(z_1-z_\l)|^2 + |\omega_2^2(z_2-z_\l)|^2}\Big)^p}{\mfr12 
|\omega_1|^{2p} + \mfr12 |\omega_2|^{2p}}\,.\nonumber
\end{eqnarray}
Now we define for $(k_1,k_2)\in\partial\H^2$ and any sequence
$(z_1,z_2)_n$ in $\H^2$ that converges to $(k_1,k_2)$, 
\begin{eqnarray}\mu_{2,p}(k_1,k_2,q,\l):=
\limsup_{(z_1,z_2)_n\to(k_1,k_2)}\mu_{2,p}(z_1,z_2,q,\l)\,. 
\end{eqnarray}
As a next step we define the beforementioned function $\mu_{3,p}$. 
First, let $\H\ni z_i\not=z_\l$ and $\boldsymbol q = (q_1,q_2)$, then
\be \mu_{3,p}(\boldsymbol z,\boldsymbol q,\l):=\sum_\s
\frac{\mathrm{cd}^p\big(\Psi(z_{\s_1},\Psi(z_{\s_2},z_{\s_3},q_2,\l),q_1,\l)
\big)}{\mathrm{cd}^p(z_1)+\mathrm{cd}^p(z_2)+\mathrm{cd}^p(z_3)}\,,
\ee
where $\s\in\Sigma:=\{(1,2,3),(2,3,1),(3,1,2)\}$ runs over the cyclic 
permutations of $(1,2,3)$. The function $\mu_{3,p}$ can be expressed in terms 
of $\mu_{2,p}$ and the auxillary function 
\beax n_j(\boldsymbol z):=\frac{\mathrm{cd}^p(z_j)}{\mathrm{cd}^p(z_1)+
\mathrm{cd}^p(z_2) + \mathrm{cd}^p(z_3)}\,,\quad j=1,2,3\,.
\eeax
Namely, 
\bea\label{mu3p}
\lefteqn{\mu_{3,p}(\boldsymbol z,\boldsymbol q,\l) =  
\sum_{\s\in\Sigma}\mu_{2,p}\big(z_{\s_1},\Psi(z_{\s_2},z_{\s_3},q_2,\l),q_1,
\l\big)}\\
&& \qquad\qquad\times\;\Big(\mfr12 n_{\s_1}(\boldsymbol z) + \mfr14 \mu_{2,p}
(z_{\s_2},z_{\s_3},q_2,\l)(n_{\s_2}(\boldsymbol z) + n_{\s_2}(\boldsymbol z)
\Big)\,.
\nonumber\eea
Then, like for $\mu_{2,p}$ above, we extend $\mu_{3,p}$ to an upper semi-continuous 
function onto the boundary of $\partial\H^3$ by taking a $\limsup$. 
Now, $\mu_{3,p}(\boldsymbol z,\boldsymbol 0,\l)\le 1-2\mu<1$ for some $\mu>0$. 
By the compactness of the boundary of $\H^3$ and the upper semi-continuity of $\mu_{3,p}$ we 
finally get the pointwise estimate $\mu_{3,p}(\boldsymbol z,\boldsymbol q,\l)\le1-
\mu$ for $\boldsymbol z$ near the boundary of $\H^3$, small $\boldsymbol q$, and 
$\l\in(-2\sqrt{2},-2\sqrt{2})$.
\end{proof}

\begin{rem} This proof is now much easier to generalize to higher branched 
trees, $T_k$, with $k\ge3$, which was first accomplished by Halasan in her
thesis~\cite{Hal}. In that case, $\Psi(z_1,\ldots,z_k,0,\l):=-1/(z_1+\cdots+z_k+\l)$
with fixed point $z_\l:=-\l/(2k) + \i\sqrt{1/k - \l^2/(4k^2)}$. Using $u =
(u_1,\ldots,u_k)$ with $u_j:=(z_j-z_\l)/\sqrt{y_j}$ and $v=(v_1,\ldots,v_k)$ 
with $v_j:=\sqrt{y_j/(y_1+\cdots+y_k)}$ we see that
$$\mathrm{cd}(\Psi(z_1,\ldots,z_k,0,\l)) =\mfr1k\,
\frac{|z_1+\cdots+z_k-kz_\l|^2}{y_1+\cdots+y_k+\Im(\l)} < 
\mfr1k\,\big|\langle u,v\rangle\big|^2\,.
$$
Therefore, by the same arguments as above and with $p>1$,
\bea\mu_{2,p}(z_1,\ldots,z_k,0,\l) &:=&\frac{\mathrm{cd}^p(\Psi(z_1,\ldots,z_k,0,
\l))}{\mfr1k\,\cd^p(z_1)+\cdots+\mfr1k\,\cd^p(z_k)} \\
&\le&\frac{\Big(\mfr1k\,\big|\langle u,v\rangle\big|^2\Big)^p}{\mfr1k\,
|u_1|^{2p}+\cdots + \mfr1k\,|u_k|^{2p}} \;\le\;1\,,\nonumber
\eea
with equality if $u=sv$ and $|u_1|=\ldots =|u_k|$. The functions $\mu_{3,p}$ 
and $n_j$ have to be changed accordingly. 
\end{rem}

In our first paper \cite{FHS1}, we attempted to construct a ``large" set of 
deterministic potentials on a (rooted) binary tree that yield ac spectrum. 
Since almost always spherically symmetric potentials cause localization we 
considered potentials that oscillate very rapidly within each sphere. The 
basic example is the following potential, $q_0$: Take vertices $v\not=w$ in the 
$n$-th sphere and $u\in S_{n-1}$ so that $(u,v)\in E$ and $(u,w)\in E$. For 
some $\d\in\R$, let $q_0(v):=\d$ and $q_0(w):=-\d$. Then continue this for every 
sphere $S_n$ except for $n=0$, where we may define $q_0(0)$ arbitrarily. 
$\l\in\R$ is in the interior of the ac spectrum of $-\D+q_0$ if and only if the 
polynomial $p(z):=z^3+2\l z^2+(2+\l^2-\d^2)z+2\l$ has two non-real, 
complex-conjugate roots and one real root. 
 
An interesting extension arises when the value $\d$ is allowed to
depend on the radius, $n$. In other words, let $\d_0>0$ be fixed and let 
$\d_1,\d_2$ be real-valued functions on $\N$. Then for vertices $v\not=w$ in 
the $n$-th sphere as above, we set $q(v):=\d_0+\d_1(n)$ and $q(w):=-\d_0+\d_2(n)$. 

\begin{prop}[\!\!\cite{FHS1}, Proposition 4.1]
\label{det two-period} Let $q$ and $q_0$ be the above potentials 
and let $\l\in\s(-\D+q_0)$. Then for $\|\d_1\|_\infty+\|\d_2\|_\infty$ small 
enough depending on $\d_0$, the Green function of $-\D+q$, $G_\l(0,0)$, is bounded. 
\end{prop}
However, these potentials (and some 
modifications thereof) are still a set of measure zero. An an explicit 
construction of a ``large" set (that is, of positive measure) remains an open 
problem.

\bigskip
In \textbf{percolation} models, one is usually interested in the occurence of infinite
clusters. A more sophisticated question is whether the spectrum of the
adjacency matrix (of the remaining graph) has an ac component. Let us start 
with the (rooted) binary tree $T_2=(V,E)$, and let $q\in[0,1)$. At every vertex 
$v\in V$, say $v\in S_n$ for some $n\in\N_0$, we delete one and only one
(forward) edge $(v,v')\in E$ or $(v,v'')\in E$ with $v',v''\in S_{n+1}$ with 
probability $q$. With probability $1-q$ we keep both (forward) edges 
$(v,v'),(v,v'')$ in the set of edges. This defines a probability measure,
$\nu_q$, on $\O:=\{0,1\}^E$, which is characterized by (we write $\o_{uv}:=
\o((u,v))$)
\begin{enumerate}
\item[(i)] $\nu_q\big(\{\o\in\O:\o_{vv'} = \o_{vv''}=1\}\big) = 1-q$ for all 
$v\in V$;
\item[(ii)] $\nu_q\big(\{\o\in\O:\o_{vv'} = 0, \o_{vv''}=1\}\big) = 
\nu_q\big(\{\o\in\O:\o_{vv'} = 1,\o_{vv''}=0\}\big) = q/2$ for all $v\in V$;
\item[(iii)] for all $u,v\in V$ with $u\not=v$ the random variables 
$(\o_{uu'},\o_{uu''})$ and $(\o_{vv'},\o_{vv''})$ are independent. 
\end{enumerate}

For every $\o\in\O$, we define the adjacency matrix of the remaining random
graph, 
\bea \D_\o:\ell^2(V)\to \ell^2(V)\,,\quad 
(\D_\o f)(v) := \sum_{u\in V} \D_\o(u,v) f(u)\,,\quad f\in  \ell^2(V)
\eea
with matrix kernel
\be \label{def:lap1}
\D_\o(u,v) :=\left\{\begin{array}{lcc} 1&\mbox{ if } \o_{uv} =1\\
0&\mbox{ otherwise }\end{array} \right.\,,\quad u,v\in V\,.
\ee 

\begin{thm} \label{thm:perc}
For every $0\le\l<2\sqrt{2}$ there exists a $q_0>0$ such that for 
all $0\le q\le q_0$, $[-\l,\l]\subseteq \s_{ac}(\D_\o)$ $\nu_q$-almost surely. 
Furthermore, the spectrum is purely ac on $(-\l,\l)$ $\nu_q$-almost surely.
\end{thm}
This particular model was suggested to one of us by Shannon Starr to whom we are
grateful.
Before we enter into some details of the proof let us start with some definitions.
For $v\in S_n\subset V$, let $\mathcal G_v=(V_v,E_v)$ be the binary graph $T_2=(V,E)$ 
truncated at $v$, that is, the largest connected subgraph of $T_2$ that 
contains $v$ but no $u\in S_k$ with $k<n$ (or simply the binary tree with root $v$); 
this truncation is different from the one in Section \ref{higher dimension}. For 
$\omega\in\O$ and $v\in V$ we define the truncated adjacency matrix, 
\be \D_\o^{(v)}:\ell^2(V_v)\to \ell^2(V_v)\,,\quad 
(\D_\o^{(v)} f)(u) := \sum_{r\in V_v} \D_\o(r,u) f(r)\,,\quad u\in V_v\,,
f\in  \ell^2(V_v)\,.
\ee 
Furthermore, for $\l\in\H$, we define the two Green functions
\bea
G(\o,\l)&:=& (\D_\o - \l)^{-1}(0,0)\,,
\\
G^{(v)}(\o,\l)&:=& (\D_\o^{(v)} - \l)^{-1}(v,v)
\eea
as the kernels of the respective resolvents. We have $G(\o,\l) = 
G^{(0)}(\o,\l)$. The recursion formula for $G^{(v)}(\o,\l)$ is
\be \label{Greenrecursion}
G^{(v)}(\o,\l) = -(G^{(v')}(\o,\l) + G^{(v'')}(\o,\l)+\l)^{-1}\,.
\ee
Finally, let
\be \rho_{\l,q}^{(v)}:=\nu_q \circ G^{(v)}(\cdot,\l)^{-1}
\ee
be the Green probability distribution defined as the image of the measure 
$\nu_q$ under the map $\o\mapsto G^{(v)}(\o,\l)$ from $\O$ to $\H$. By 
translation-invariance, the measure $\rho_{\l,q}^{(v)}$ does, in fact, not 
depend on $v$, and we shortly write $\rho_{q}$ by also suppressing the spectral
parameter $\l$. 

\begin{proof}[Sketch of proof of Theorem \ref{thm:perc}] Using the weight function 
$\mathrm{cd}$ from \eqref{def:cd2} with the same $z_{\l}$ and $p>1$ we define 
the moment
\be M_p(\rho_{q}) := \int_\H \mathrm{cd}^{p}(z)\, d\rho_{q}(z) \,.
\ee
Applying the recursion relation \eqref{Greenrecursion} and the symmetry between
the variables $z_1$ and $z_2$ below we have
\be M_p(\rho_{q})=\int_{\H^2} \mu_{2,p,q}(z_1,z_2,\l)\,\big[\mfr12\,
\mathrm{cd}^{p}(z_1) + \mfr12\,\mathrm{cd}^{p}(z_2)\big] \,d\rho_{q}(z_1) 
d\rho_{q}(z_2)
\ee
with $\mu_{2,p,q}(z_1,z_2,\l):=[{q} \,\mathrm{cd}^{p}(-1/(z_1+\l)) + (1-q) \,
\mathrm{cd}^{p}(-1/(z_1+z_2+\l)^{-1})]$. Then, as in \eqref{mu2p}, we apply
once more the recursion relation and write the result in the form
\bea \lefteqn{M_p(\rho_{q})=}\\
&&\int_{\H^3} \mfr13\mu_{3,p,q}(z_1,z_2,z_3,\l)\,\big[\mathrm{cd}^p(z_1) + 
\mathrm{cd}^{p}(z_2)+\mathrm{cd}^{p}(z_3)\big] 
\,d\rho_{q}(z_1) d\rho_{q}(z_2)d\rho_{q}(z_3)\nonumber\,.
\eea
The function $\mu_{3,p,q}(z_1,z_2,z_3,\l)$ is expanded as a function of $q$ so
that
\be \mu_{3,p,q} = (1-q)^2 \mu_{3,p} + q R\,,
\ee
where $\mu_{3,p}$ is the function in \eqref{mu3p} with $q_1=q_2=0$ and $|R|\le C_K$ on 
$\H^3\setminus K$ for a compact set $K$. For $q$ small enough we achieve that
$(1-q)^2 \mu_{3,p} + q R \le (1-\mu/2)$ outside such a compact set $K$ with 
$\mu>0$. Hence, $M_p(\rho_{q})\le C/(1-\mu/2)$.
\end{proof}

\begin{rems}
\begin{enumerate}
\item[(i)]
We do not know the full spectrum of the adjacency matrix, $\D_\o$,
nor do we have information on the remaining (point) spectrum. 
\item[(ii)]
In this percolation model, there is always an infinite cluster even when $q=1$.
This is in contrast to the genuine bond-percolation tree model, where an edge 
is deleted with probability $q$ independently of other edges. Here, the
percolation threshold for the existence of an infinite cluster is $q_c=1/2$, see
\cite{Lyons}.
This model seems harder to analyze, at least
from the standpoint of our method. The reason is that the point spectrum is
dense in the full spectrum of the random percolation graph since almost surely 
there are arbitrarily large subtrees disconnected from the random graph for 
which the spectrum lies inside $(-2\sqrt{2},2\sqrt{2})$. Thus there is no
interval of pure ac spectrum if it happens to exist at all. Besides, we are not
aware of a conjectured value for a critical (quantum percolation) value 
$q_{\rm{qp}}$ up to which the adjacency matrix has an ac component; 
$q_{\rm{qp}}\le 1/2$ since an infinite cluster is required to exist.
\end{enumerate}
\end{rems}

\section{Strongly correlated random potential on a tree}
\label{strongly correlated}

There is a large gap between the known results for the tree and the open 
problem on $\Z^d$ for $d\ge3$. Therefore it seems worthwhile to address some of 
the problems that would come up on $\Z^d$ in simpler toy models. In order to 
see a strong effect of correlations we consider a transversely 2-periodic 
random potential. The potential is defined by choosing two values $\boldsymbol q
= (q_1,q_2)$ of the potential at random, independently for each sphere in the 
tree. These two values are then repeated periodically across the sphere and
hence the potential is strongly correlated. Such a two-periodic potentials 
can exhibit either dense point spectrum or absolutely continuous spectrum
depending on the correlations of $q_1$ and $q_2$.

We will prove that if the values of $q_1$ and $q_2$ are sufficiently 
uncorrelated (see assumption (\ref{nuassn2}) below) then there will be some 
ac spectrum, as is the case for the iid Anderson model. However, since in some 
sense this model is so close to being one-dimensional, the proof has some 
features not appearing in the tree model of Section \ref{tree}. This time, 
the proof follows from an estimate of an \emph{average} over potential values 
$\boldsymbol q$ of functions $\mu(z,\boldsymbol q)$, similar in both models, 
that measure the contraction of a relevant map of the plane. We seek an 
estimate of the form $\int \mu(z,\boldsymbol q)\, d\nu_a(\boldsymbol q) < 1$ for 
$z$ near the boundary of $\H$ at infinity. In the proof of Theorem \ref{ext states conj}
we have used the independence of the potentials across the sphere in proving that
$\mu(z,\boldsymbol 0)$ is already less than one. Then small values of 
$\boldsymbol q$ in the integral are handled by semi-continuity. In the present
situation, $\mu(z,\boldsymbol 0)=1$ and perturbations in $\boldsymbol q$ send 
it in both directions. Thus we must use cancellations in the integral
over $\boldsymbol q$ in an essential way. 

Our method extends to the case where the joint distributions are not identical, as long
as they are all centered and satisfy certain uniform bounds. This is 
significant since in this case we lose the self-similarity that has been used 
in previous proofs.  

We make the following assumptions about the measure $\nu$. First, it has 
compact support, and for simplicity,
\be\label{nuassn1}
\nu \mbox{ is supported in } \{\boldsymbol q = (q_1,q_2) : |q_1|\le 1, 
|q_2|\le 1\}\,.
\ee
Then, the measure is centered on zero:
\be\label{nuassn3}
\int_{\R^2} (q_1 + q_2) \,d\nu(\boldsymbol q) = 0\,.
\ee
Let $c_{ij} := \int_{\R^2} q_iq_j \,d\nu(\boldsymbol q)$. Then finally, 
\be\label{nuassn2}
c:=c_{11}+c_{22} >0 \quad\mbox{ and }\quad \delta := \frac{2c_{12}}{c_{11}+
c_{22}} < 1/2\,.
\ee
The first inequality in \eqref{nuassn2} simply says that $\boldsymbol q$ is not
identically zero. The second is a bound on the correlation. Completely 
correlated potentials (that is, the one-dimensional case where the spectrum is 
localized) would correspond to $\delta = 1$.


We have proved the following theorem.
\begin{thm}[\!\!\cite{FHS3}, Theorem 2]\label{twoperiodic}
Let $\nu_{(0)}$ be a probability measure of bounded support for the potential 
at the root, let $\nu$ be a probability measure on $\R^2$ satisfying 
\eqref{nuassn1}, \eqref{nuassn3} and \eqref{nuassn2} and let $H_a$ be the random 
discrete Schr\"odinger operator on the binary tree corresponding to the 
transversely two-periodic potential defined by the scaled measure 
$\nu_a$. There exists $\lambda_0 \in (0,2\sqrt{2})$ such that for sufficiently 
small $a$ the spectral measure for $H_a$ corresponding to $\delta_0$ has 
purely ac spectrum in $(-\lambda_0,\lambda_0)$.
\end{thm}

\begin{rem} When the random variables $q_1$ and $q_2$ are independent, that is,
when $\d=0$, our proof shows that $\l_0$ can be chosen to be 2. The 
determination of the maximum $\l_0$ remains an open problem.
\end{rem}

\section{Loop tree models}\label{loop tree}

There are several interesting ways to add loops to a tree which are
sometimes called decorated trees. Here we present three possibilities of 
adding new edges that connect vertices inside the same sphere.

In our first attempt we connect each vertex $(n,2^i)$ inside each sphere $S_n$ 
with $(n,2^{i}+1)$ and $(n,2^{i}-1)$ modulo $2^n$. That amounts to adding to the
adjacency matrix of the tree the adjacency matrices of the nearest neighbor 
chains $\{0,1,\ldots,2^n-1\}$ with periodic boundary conditions. We call this 
the regular loop tree model. Every vertex other than the root has five neighbors. In the next 
subsection we present the derivation of the fixed point equation that determines 
the spectrum at the root.  Finding the spectrum of this new adjacency matrix turns 
out to be difficult and remains an open problem. 

In a second attempt we modify these new connections to mean-field connections.
This new mean-field Laplacian can be solved explicitly so that we can take on
the next step and add a random potential. Here we limit ourselves to a
special case, namely to a two-periodic Bernoulli random potential that we have 
studied in the previous section. We present the model and the main result in 
Subsection \ref{sect:mean-field}. Proving ac spectrum for the Anderson model 
(that is, with iid random potential) on this mean-field tree model is still
an open problem.  

The third loop tree model was suggested to us by Laszlo Erd\"os. In its simplest
version, one adds to each sphere of the tree a single loop (of weight $\g$) 
that connects two arbitrarily chosen sites within a sphere. It would be 
interesting to prove the (in)stability of the ac spectrum for small $\g>0$.

\subsection{Regular loop tree model}

Each vertex $v$ in the $n$-th sphere of the binary tree, $T_2=(V,E)$, is of the
form $v=(n,j)$ with $0\le j\le 2^n-1$. We now also call $v,w\in S_n$ nearest 
neighbors if $w=(n,j\pm1\mod 2^n)$. The newly added edges are denoted by
$E^{\mathrm{rlt}}$. In order to compare with the usual adjacency matrix of the
tree we introduce a parameter $\g$ that puts the weight $\g$ on the new 
connections inside a sphere. The new adjacency matrix, $\D_\g$, is now defined
by the kernel 
\be \D_\g(v,w):=\left\{\begin{array}{cl}1&\mbox{ if }(v,w)\in E\\\g&\mbox{ if }
(v,w)\in E^{\mathrm{rlt}}\\0&\mbox{ else }\end{array}\right.\,.
\ee
Furthermore, let $\g D:=\D_\g-\D$, and let $D_n$ be $D$ restricted to $S_n$. 

For $n\in\N_0,N:=2^n$, and $\l\in\H$ we consider the generalized M\"obius
transformations $\Phi_n:\SH_{2N}\to \SH_{N}, \Phi_n(Z) := -(E_n^T Z E_n + \g 
D_n +\l )^{-1}$ between the respective Siegel half-spaces. When $\g\not=0$ 
then diagonal matrices are no longer mapped to diagonal matrices. An invariant 
subset of matrices that is preserved under this flow is the set of 
\emph{circulant} (or Toeplitz) matrices. Recall that an $N\times N$ matrix $Z$ 
is called circulant if $Z_{i,j} = z_{(j-i)\mathrm{ mod }N}$. That is,
$$ Z = \left[\begin{array}{ccccc}z_0&z_1&z_2&\cdots&z_{N-1}\\z_{N-1}&
       z_0&z_1&\cdots&z_{N-2}\\z_{N-2}&z_{N-1}&z_0&\cdots&z_{N-3}\\
       \vdots&\vdots&\vdots&\ddots&\vdots\\
       z_1&z_2&\cdots&z_{N-1}&z_0\end{array}
       \right]\,.
$$
Circulant matrices are characterized by the condition that they commute with 
the shift operator. Therefore we can diagonalize circulant matrices by the 
finite Fourier transform. 
The finite Fourier transform, $U_n\in\mathrm{Mat}(N,\C)$, is defined as
\be (U_n)_{j,k} := N^{-1/2} \,\e^{2\pi \i jk/N}\,, \qquad
j,k=0,1,\ldots,N-1\,,\; N=2^n\,.
\ee
Here are some simple properties.
\begin{lem} 
\begin{enumerate} 
\item Let $Z$ be an $N\times N$ circulant matrix with first row
$\boldsymbol z=[z_0,z_1$, $\ldots,z_{N-1}]$. For $j=0,1,\ldots,N-1$, let 
$f^{(n)}_j := \sum_{\ell=0}^{N-1} z_\ell \,\e^{2\pi \i \ell k/N}$. Then 
\be (U_n^* Z U_n)_{j,k} = \d_{jk}\,f^{(n)}_j\,, \quad j,k=0,1,\ldots,N-1 \,.
\ee
In particular, for the spherical Laplacian $D_n$ we have 
\be (U_n^* D_n U_n)_{j,k} = 2\,\d_{j,k} \cos\big(\mfr{2\pi j}{N}\big)\,, 
\quad j,k=0,1,\ldots,N-1\,.
\ee
\item For $j=0,1,\ldots,N-1,\ell=0,1,\ldots,2N-1$ we have
\be (U_{n}^* E^T_n U_{n+1})_{j,k} = 2^{-1/2} \big(1+\e^{2\pi \i k/2^{n+1}}\big)
   \,(\d_{j,k} +\d_{j+2^n, k})\,.
\ee
\item For $j,k=0,1,\ldots,N-1$ we have
\bea\label{7.6}
\lefteqn{\Big(U_n^*(E^T_n Z E_n + \g D_n + \l)^{-1} U_n \Big)_{j,k}}
\\&& = \d_{j,k}\,
    \frac{1}{2 \cos^2\big(\mfr{2\pi k}{2^{n+2}}\big) \,f^{(n+1)}_{k/2} 
    + 2 \sin^2\big(\mfr{2\pi k}{2^{n+2}}\big) \,f^{(n+1)}_{k/2+2^n} 
    + 2\g \cos\big(\mfr{2\pi k}{2^{n}}\big) + \l} \,.
\nonumber\eea
\end{enumerate}
\end{lem}
\begin{proof} This is all quite easy but nevertheless ... 
\beax (U_n^* Z U_n)_{j,k} &=& N^{-1} \sum_{m,\ell=0}^{N-1} 
      \e^{-2\pi\i (jm-\ell k)/N} \, z_{(\ell-m)\mathrm{ mod }N}
\\
&=&\sum_{\ell=0}^{N-1} z_\ell \,\e^{2\pi\i \ell/N}\, N^{-1}
   \sum_{m=0}^{N-1} \e^{-2\pi\i m (j-k)/N}       
\,=\,\d_{j,k}\,f^{(n)}_j  \,.  
\eeax
In a similar vein we obtain
\beax (U_{n}^* E^T_n U_{n+1})_{j,k} &= &\sum_{\ell=0,1,\ldots, N-1, \atop
       \ell' = 0,1,\ldots, 2N-1} U^*_{j\ell}\, E^T_{\ell,\ell'} 
       \,U_{\ell' k}
\\
&=& 2^{-n-1/2} \,\sum_{\ell,\ell'} \e^{-2\pi\i j\ell/N}\,
    \big(\d_{2\ell,\ell'} + \d_{2\ell+1,\ell'}  \big) \,
    \e^{-2\pi\i\ell' k/2N} 
\\
&=& 2^{-1/2} \big(\d_{j,k} + \d_{j+N,k}\big) +
    2^{-1/2} \,\e^{2\pi\i k/2N}\,\big(\d_{j,k} + \d_{j+N,k}\big)\,.
\eeax
The third claim follows from the first two by noticing that $U_n^*(E^T_nZE_n+
\g D_n +\l)^{-1} U_n = (U_n^*E^T_nU_{n+1} U_{n+1}^* Z U_{n+1} (U_{n}^* E_n 
U_{n+1})^* + \g\, U_n^* D_n U_n + \l)^{-1}$. This shows that
\beax \lefteqn{\big(U_n^*E^T Z_{n+1} E U_n\big)_{j,k}}
\\ 
&=& \mfr12 \d_{j,k} \Big[ \big(1+\e^{2\pi\i j/2^{n+1}}\big)\, f_j^{(n+1)} \,
    \big(1+\e^{-2\pi\i j/2N}\big)  
\\
&&\phantom{\mfr12 \d_{j,k} \Big[}     
    +\,\big(1-\e^{2\pi\i j/2N}\big)\, 
    f_{j+N}^{(n+1)} \,\big(1-\e^{-2\pi\i j/2N}\big) \Big]
\\
&=&\d_{j,k} \,\Big[\big(1+\cos(\mfr{2\pi j}{2N})\big) \, f_j^{(n+1)} \,
   + \big(1-\cos(\mfr{2\pi j}{2N})\big) \, f_{j+N}^{(n+1)}\Big]
\\
&=&2 \,\d_{j,k} \Big[\cos^2(\mfr{2\pi j}{4N})\,f_j^{(n+1)}  + 
     \sin^2(\mfr{2\pi j}{4N}) f_{j+N}^{(n+1)}\Big]\,.
\eeax     
\vskip-2em
\end{proof}
\eqref{7.6} implies that if $Z\in \SH_{2N}$ is circulant with Fourier
transformation $f^{(n+1)}$ then $\Phi_n(Z)\in \SH_{N}$ is circulant with 
Fourier transformation $f^{(n)}$, and so that 
\be \label{equ: n-fixed point}
    f_k^{(n)} = - \frac{1}{2 \cos^2\big(\mfr{\pi k}{2N}\big) 
    \,f^{(n+1)}_{k/2} 
    + 2 \sin^2\big(\mfr{\pi k}{2N}\big) \,f^{(n+1)}_{k/2+N} 
    + 2\g \cos\big(\mfr{2\pi k}{N}\big) + \l} \,. 
\ee
Letting $n\to\infty$ and setting $f(\mfr{\pi k}{2N}):=f^{(n)}_k$ we obtain 
the fixed point equation,
\bea \label{equ: fixed point}
    f(\theta) &=& - \frac{1}{2 \cos^2\big(\mfr{\theta}{4}\big) 
    \,f(\mfr{\theta}{2}) + 2 \sin^2\big(\mfr{\theta}{4}\big) 
    \,f(\mfr{\theta}{2}+\pi) + 2\g \cos(\theta) + \l}\,,  
\eea
for functions $f:[0,2\pi]\to\H$. [For $\g=0$ and $\Im(\l)>0$ the only solution
to \eqref{equ: fixed point} is the constant function with value $z_\l$ from 
\eqref{fixpnt}.]

The truncated Green functions $Z=Z_n\in \SH_N$ are further restricted by the 
condition that $Z_n$ has to be symmetric (not hermitean). This implies that 
the first row, $\boldsymbol z=[z_0,z_1,\ldots, z_{2^{n}-1}]$ of $Z_n$
is symmetric with respect to the middle co-ordinate, $2^{n-1}$. 
That is, 
\be \label{Z:symm} 
    \boldsymbol z=[z_0,z_1,\ldots,z_{2^{n-1}-1},z_{2^{n-1}},z_{2^{n-1}-1},
    \ldots,z_1]
    \,.
\ee
Therefore, $f^{(n+1)}_{2^n-k} = f^{(n+1)}_{2^n+k}$ and consequently, 
$f(\pi-\theta) = f(\pi+\theta)$.    

\bigskip
The only place where we are able to evaluate the solution of 
\eqref{equ: fixed point} explicitly is for $\theta\in\{0,2\pi\}$, where we 
find that $G_\l(0,0) = f(0) = f(2\pi) = \frac{2\g+\l}4 + \frac{\i}{4}
\sqrt{8-(2\g+\l)^2}$. For $f(0)$ to be in $\H$ we get the condition that 
$|2\g+\l|< 2\sqrt{2}$. Therefore, $[-2\g-2\sqrt{2},-2\g+2\sqrt{2}]$ is in the ac
spectrum of $\D_\g$. On the other hand, let $\phi_n\in\ell^2(V)$ with
$\phi_n(v):=2^{-n/2}$ for $v\in S_n$ and zero otherwise. Then, the variational
energy, $\langle \phi_n, \D_\g \phi_n\rangle = 2\g$. So for large $\g$, this
energy is outside the interval $[-2\g-2\sqrt{2},-2\g+2\sqrt{2}]$ and thus, 
unlike for $\g=0$, $f(0)$ does not alone determine the full spectrum.

\subsection{Mean-field loop model}\label{sect:mean-field}

Now we add a weighted complete graph to every sphere in the binary tree. 
Since the weights are chosen to make the total added weights the same in each 
sphere, this is a sort of mean-field model. Pick a number $\gamma >0$. Each 
added edge (dotted line in the figure below) in the $n$-th sphere $S_n$ is 
given the weight $\gamma 2^{-n}$. That is, we define the adjacency matrix, 
$\D_\g^{\mathrm{mf}}$, through the kernel
\be \D_\g^{\mathrm{mf}}(v,w):=\left\{\begin{array}{cl}1&\mbox{ if }(v,w)\in 
E\\\g2^{-n}&\mbox{ if }v,w\in S_n\\0&\mbox{ else }\end{array}\right.\,.
\ee
\begin{figure}[h]
  \centering
  \includegraphics[width=0.3\textwidth]{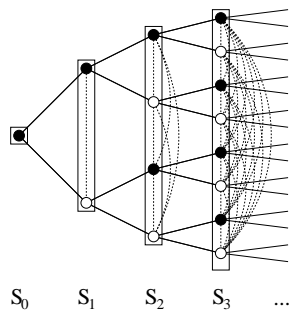}
  \caption{Rooted binary tree with mean-field edges insides spheres and 
transversely 2-periodic potential.}
  \label{fig:}
\end{figure}
We call the new (weighted) graph the mean-field binary tree. The spectrum of 
the mean-field adjacency matrix, $\D_\g^{\mathrm{mf}}$, is the union of two 
intervals $[-2\sqrt{2}+\gamma,2\sqrt{2}+\gamma]\cup[-2\sqrt{2},2\sqrt{2}]$ and 
is purely ac. This can be seen by using a Haar basis~\cite{FHS3}.

For simplicity, we considered a random potential that is transversely 
two-periodic and defined by the product of two independent Bernoulli 
measures for $q_1$ and $q_2$,
\be
d\nu(q_1,q_2) = \frac14\big(\delta(q_1-1) + \delta(q_1+1)\big)\big(\delta(q_2-1) 
+ \delta(q_2+1)\big)\,.
\ee
Then we have the following theorem.
\begin{thm}[\!\!\cite{FHS3}, Theorem 9]
\label{meanfield} Let $\nu_{(0)}$ be a probability measure of 
bounded support for the potential at the root and $\nu$ be the product of 
Bernoulli measures defined above and let $H_{a,\gamma}:=-
\D_\g^{\mathrm{mf}} + a\,q$ be the random discrete Schr\"odinger operator on 
the mean-field binary tree corresponding to the transversely two-periodic 
potential defined by the scaled distribution $\nu_a$ and weight $\gamma$. 
There exist $0 < \lambda_0, \lambda_1 < 2\sqrt{2}$ such that for sufficiently 
small $a$ the spectral measure for $H_a$ corresponding to $\delta_0$ has 
purely ac spectrum in $\{\l: |\l|\le \l_0,|\l - \gamma|\le\l_1\}$.
\end{thm}

In this theorem, the constant $\lambda_0$ has the same value as in Theorem
\ref{twoperiodic}, while $\lambda_1$ can be taken to be any positive number 
less than $2\sqrt{2}$.

\subsection*{Acknowledgment}
WS is indepted to Florian Sobieczky for organizing the wonderful Alp--workshop 
in St. Kathrein. We are also grateful to the referee for many useful comments.
\end{document}